\documentclass[a4paper,cleveref,autoref]{lipics-v2021}

\usepackage{graphicx}
\graphicspath{{figures/}}
\usepackage{fancybox}
\usepackage{stmaryrd}
\usepackage{xspace}
\usepackage{stmaryrd}

\usepackage{xcolor}
\definecolor{blu3}{rgb}{.1,.0,.4}
\usepackage{hyperref}
\hypersetup{colorlinks=true, linkcolor=blu3, urlcolor=blu3, citecolor=blu3, pdfpagemode=UseNone, pdfstartview=} %colored links, no rectangles, no bookmarks, zoom

\newtheorem*{OVC}{Orthogonal Vectors Conjecture}
\newtheorem*{SETH}{Strong Exponential Time Hypothesis}

\newcommand{\RR}{\ensuremath{\mathbb R}}  % real numbers
\newcommand{\ZZ}{\ensuremath{\mathbb Z}}  % integer numbers
\newcommand{\Hout}{\ensuremath{H_{\rm out}}}  % H_out

\def\DEF#1{\textbf{\emph{#1}}}
\def\eps{\varepsilon}

\DeclareMathOperator{\interior}{int}
\DeclareMathOperator{\Holes}{holes}
\DeclareMathOperator{\tw}{tw}
\DeclareMathOperator{\diam}{diam}

\newcommand{\dart}[2]{#1 \shortrightarrow #2}

\definecolor{defblue}{rgb}{0.1,0.4,0.6} % bit darker
\let\emph\relax\DeclareTextFontCommand{\emph}{\color{defblue}\em}

\newcommand{\DSinit}{\ensuremath{\textsc{Init}}}
\newcommand{\DSadd}{\ensuremath{\textsc{Add}}}
\newcommand{\DSmin}{\ensuremath{\textsc{Min}}}

\title{Testing Whether a Subgraph is Convex or Isometric}
\author{Sergio Cabello}{Faculty of Mathematics and Physics, University of Ljubljana, Ljubljana, Slovenia \and Institute~of~Mathematics, Physics and Mechanics, Ljubljana, Slovenia}{sergio.cabello@fmf.uni-lj.si}{0000-0002-3183-4126}{}

\authorrunning{Cabello} 
\Copyright{Sergio Cabello}


\hideLIPIcs

\relatedversion{A preliminary version of this work appeared 
at the 19th Algorithms and Data Structures Symposium (WADS 2025)~\cite{Cabello25}.} %optional, e.g. full version hosted on arXiv, HAL, or other respository/website

\funding{Funded in part by the Slovenian Research and Innovation Agency (P1-0297, N1-0218, N1-0285, J1-70045).
Funded in part by the European Union (ERC, KARST, project number 101071836). Views and opinions expressed are however those of the authors only and do not necessarily reflect those of the European Union or the European Research Council. Neither the European Union nor the granting authority can be held responsible for them.}%optional, to capture a funding statement, which applies to all authors. Please enter author specific funding statements as fifth argument of the \author macro.

\acknowledgements{I am very grateful to Vladimir Batagelj and Sandi Klavžar 
	for discussing the related work. I am also very grateful to anonymous reviewers 
	that provided several relevant references and spotted an error in a previous version.}%optional

\nolinenumbers %uncomment to disable line numbering

\begin{document}

\maketitle

\begin{abstract}
	We consider the following two algorithmic problems:
	given a graph $G$ and a subgraph $H\subseteq G$, 
	decide whether $H$ is an isometric or a geodesically convex subgraph of $G$.
	It is relatively easy to see that the problems can be solved by computing
	the distances between all pairs of vertices.
	We provide a conditional lower bound showing that, 
	for sparse graphs with $n$ vertices and $\Theta(n)$
	edges, we cannot expect to solve the problem in $O(n^{2-\eps})$ time 
	for any constant $\eps>0$. We also show that the problem can be solved
	in subquadratic time for planar graphs and in near-linear time for graphs of 
	bounded treewidth. Finally, we provide a near-linear time algorithm for the
	setting where $G$ is a plane graph and $H$ is defined by a few cycles in $G$.
\end{abstract}

%%%%%%%%%%%%%%%%%%%%%%%%%%%%%%%%%%%%%%%%%%%%%%%%%%%%%%%%%%%%%%%%%%%%%%%%%%%%%%%
\section{Introduction}

Isometric and geodesically convex subgraphs\footnote{There are several distinct 
concepts of convexity in graphs~\cite{ChangatMS05,Pel13}. Since in this work we only talk about geodetic 
convexity, we will drop the adjective ``geodesically''.} are two natural and fundamental 
concepts in the area of metric graph theory~\cite{BC08,Pel13}. 
Our objective in this paper is to study the algorithmic problem of recognizing whether 
a given subgraph is convex or isometric. 

\subparagraph{Setting.}
Let $G$ be an undirected graph with abstract, positive edge 
lengths $\ell:E(G)\rightarrow \RR_{>0}$.
We call $G$ the \DEF{host graph} and assume henceforth that it is \emph{connected}.
The \DEF{length} of a walk (or path) $\pi$ in $G$, denoted by $\ell(\pi)$, 
is the sum of the lengths of its edges; thus $\ell(\pi)= \sum_{e\in E(\pi)} \ell(e)$,
where the sum is with multiplicity.
For any two vertices $x,y\in V(G)$, a \DEF{shortest path} from $x$ to $y$ 
is a minimum-length path connecting $x$ and $y$,
and the \DEF{distance} between $x$ and $y$ in $G$, denoted by $d_G(x,y)$,
is the length of a shortest path from $x$ to $y$.
  
For any two vertices $x,y$ of $G$, the \DEF{interval} $I_G(x,y)$ is the subgraph 
of $G$ defined by the union of all shortest paths from $x$ to $y$. Formally
\begin{align*}
	V(I_G(x,y)) &= \{ u\in V(G)\mid d_G(x,u)+d_G(u,y)= d_G(x,y) \}.\\
	E(I_G(x,y)) &= \{ uv\in E(G)\mid d_G(x,u)+\ell(uv)+ d_G(v,y)= d_G(x,y) \text{ or }\\
			   &\phantom{\{ uv\in V(G)\mid}~~~~~ d_G(x,v)+\ell(uv)+ d_G(u,y)= d_G(x,y) \}.
\end{align*}
A subgraph $H$ of $G$ is (geodesically) \DEF{convex} (in $G$) if and only if
\[
	\forall x,y\in V(H):~~ I_G(x,y)\subseteq H,
\]
and it is \DEF{isometric} (in $G$) if and only if
\[
	\forall x,y\in V(H):~~ d_G(x,y) = d_H(x,y).
\]
We consider the following algorithmic problems:
\begin{itemize}
	\item Given a host graph $G$ and a subgraph $H\subseteq G$, decide whether $H$ is an isometric subgraph.
	\item Given a host graph $G$ and a subgraph $H\subseteq G$, decide whether $H$ is a convex subgraph.
\end{itemize}

In graph theory, the concepts of isometric and (geodesically) convex subgraphs 
are defined for the setting where $\ell(e)= 1$ for all $e\in E(G)$.
In this unit-length setting, only induced subgraphs are considered because
non-induced graphs cannot be isometric nor convex.
In the general case with edge lengths, we may have $d_G(x,y)< \ell(xy)$ for some edges $xy$ of $G$.

Henceforth, we use $n$ and $m$ for the number of vertices and edges of the host graph $G$.

\subparagraph{Related work.}
Dourado et al.~\cite{DouradoGKPS09} provide an efficient algorithm to test convexity
in the case of unit-length edges.
Their approach can be easily adapted to graphs with arbitrary edge lengths, as follows.
For each vertex $x$ of the subgraph $H$, one computes the graph
$H_x=\bigcup_{y\in V(H)}I_G(x,y)$. The subgraph $H$ is convex in $G$ if and only if 
$H_x$ is contained in $H$ for all $x\in V(H)$.
The computation of $H_x$ takes $O(m)$ time once we have a shortest-path tree in $G$ 
from $x$, and testing whether $H_x\subseteq H$ also takes $O(m)$ time. 

We conclude that testing whether a subgraph is convex can be done
by computing a shortest-path tree from each vertex, the so-called
all-pairs shortest-path (APSP) problem, followed by tests that take $O(nm)$ time.
With a simple modification, the same time bound holds for testing whether the given 
subgraph is isometric. The running time to solve the APSP problem depends 
on the assumptions about the edge lengths:
\begin{itemize}
	\item for unit edge lengths, the APSP problem is solved in $O(nm)$ time 
	using a simple BFS;
	\item when the edge lengths are natural numbers and multiplications take
	constant time, Thorup~\cite{Thorup99} solves the APSP problem in $O(nm)$ time
	using the word RAM model of computation;
	\item if we only perform additions and comparisons of edge lengths,
	then the algorithm of Pettie and Ramachandran~\cite{PettieR02} for APSP
	takes $O(nm\,\alpha(m,n))$ time. Using Dijkstra's algorithm with 
	Fibonacci heaps~\cite{FredmanT87}, the running time is slightly worse, $O(n^2\log n + nm)$.
\end{itemize}
We conclude that testing whether a subgraph is convex or isometric can be done
in $\tilde O(nm)$ time in all situations. 

Convex subgraphs play an important role in the theory of \emph{median 
graphs}~\cite{Chepoi1988,MulderS79,Mulder78,Nebesky71}, a class of graphs that has received
substantial attention~\cite{BC08,HammackIK,KlavzarM99}.
For median graphs, convex sets and so-called \emph{gated sets} are the same concept;
see for example~\cite[Lemma 3.3]{BeneteauCCV22}. This can be used to test convexity in 
linear time: it suffices to check whether 
every vertex outside $H$ has at most one neighbor in $H$.
Imrich and Klav{\v z}ar~\cite{ImrichK98} provided a characterization of convex 
subgraphs in bipartite graphs in terms of $\Theta$-classes.
The characterization is for unit edge lengths and leads to an algorithm taking $O(nm)$ time.
They used it to obtain a fast algorithm to recognize median graphs. 
See also~\cite{GlantzM17}.
It was later shown that the recognition of median graphs
is closely related to the detection of triangles in graphs~\cite{ImrichKM99}.
If the host graph $G$ is the Cartesian product of graphs, 
$G=G_1\square\dots \square G_k$, and the edges have unit length,
a subgraph $H\subseteq G$ is convex if and only if $H=H_1\square\dots \square H_k$,
where each $H_i$ is a convex subgraph of $G_i$; see~\cite[Lemma 6.5]{HammackIK}.
Together with linear-time algorithms for factorization of graphs~\cite{ImrichP07}, 
this leads to an efficient algorithm for Cartesian products of 
substantially smaller graphs.
Convex subgraphs are being considered also in more applied 
areas; see for example~\cite{GavrilevM23,MarcS18,SubeljFCK19,TG21}.

A graph is a \emph{partial cube} \cite{Djokovic1973,Firsov65,Winkler84} 
if it is isomorphic to an isometric subgraph of some hypercube $Q_d$. 
Partial cubes enjoy a rich structure and appear in several 
different combinatorial settings; see for example~\cite{BC08,HammackIK,Ovchinnikov2011}. 
Eppstein~\cite{Eppstein11} provided an algorithm to test in $O(n^2)$ time 
whether a given graph with $n$ vertices is a partial cube.
Note that this setting is different to ours because in that setting
the input is not given as a subgraph of $Q_d$, but have to test whether 
the given graph is isomorphic to an isometric subgraph of $Q_d$. 
In our setting, in contrast, we are given the host graph $G$ and 
the subgraph $H$ so that we can immediately identify each vertex or edge
of $H$ in $G$. 
The concept of isometric embeddings in the hypercube or other products of graphs
has also been studied recently for weighted graphs; see~\cite{BerleantSCWB23,SheridanBBCW23}.

\subparagraph{Our contribution.}
We first provide a series of characterizations of isometric (\cref{sec:isometric_criteria})
and convex subgraphs (\cref{sec:convex_criteria}).
Most notably, one of the characterizations uses a marked version of the 
of the so-called Wiener index, which we will define below.
We think that the characterizations are new and of independent interest.
Although proving each of our characterizations is not difficult, identifying them
is closely related to the development of our algorithms.

Our first result regarding algorithms is the following conditional lower bound: 
if $G$ has $n$ vertices and $\Theta(n)$	edges, then there is no algorithm to test
whether a given subgraph $H$ of $G$ is isometric
or convex in $O(n^{2-\eps})$ time, for any constant $\eps>0$,
unless the Orthogonal Vectors Conjecture (OVC) and the Strong
Exponential Time Hypothesis (SETH) fail.
This means that for general graphs we cannot expect to test whether a subgraph is
convex or isometric in, say, $O(nm^{0.999})$ time, unless some standard
assumptions are wrong.
The lower bound holds for graphs with unit edge lengths and implies that
the algorithms mentioned above, with running times $\tilde O(nm)$, 
are asymptotically optimal up to polynomial factors $\Theta(n^\eps)$ for any $\eps>0$.
Our reduction is via the diameter problem. 
We review OVC and SETH in \cref{sec:lower}, where we also 
provide our reduction via the diameter problem.

The characterization using the marked version of the Wiener index can be 
combined with adaptations of known algorithms to obtain the following results
about testing whether a given subgraph is isometric or convex:
\begin{itemize}
	\item If the host graph is planar, 
		the test can be performed in $\tilde O(n^{5/3})$ time.
	\item If the host graph has treewidth bounded by a constant $k\ge 3$, 
		the test can be performed in $O(n \log^{k-1} n)$ time; the constant
		hidden in the $O$-notation depends on $k$.
		When the treewidth $k$ is considered a parameter, 
		the test can be performed in $n^{1+\eps} 2^{O(k)}$, for any constant $\eps>0$;
		the constant hidden in the $O$-notation depends on the choice of $\eps$.
	\item If the host graph is unweighted and has no $K_h$-minor, for a fixed $h$,
		then we can test whether a subgraph is isometric in subquadratic time.
		Note that this result does not hold for testing convexity nor for
		the case of weighted graphs.
\end{itemize}

Finally, in \cref{sec:plane} we consider the following setting.
The host graph $G$ is a \DEF{plane graph}, that is, a planar graph with a fixed embedding
in the plane. Let $C$ be a cycle in $G$ and let $H=\interior(G,C)$ be
the subgraph of $G$ bounded by $C$ and including $C$; see~\cref{fig:example} left.
We show how to decide whether $H$ is an isometric or a convex subgraph of $G$ in $O(n\log n)$ time.
The result can be extended to subgraphs of $G$ defined by a constant number of cycles;
see the right of~\cref{fig:example} for an example and \cref{sec:plane} for more precise
definitions.
Compared to the other results presented in the paper, this result is interesting
because it uses different characterizations and a different algorithmic paradigm.
We use multiple-source shortest-path (MSSP) algorithms for planar 
graphs~\cite{CabelloCE13,Klein05} and combine the distances in the interior 
and the exterior of $C$, as the source of the shortest-path tree moves
along the cycle $C$. As far as we know, this is the first application of MSSP where 
we maintain a combination of the interior and the exterior distances.

\begin{figure}
\centering
	\includegraphics[page=2,width=.9\textwidth]{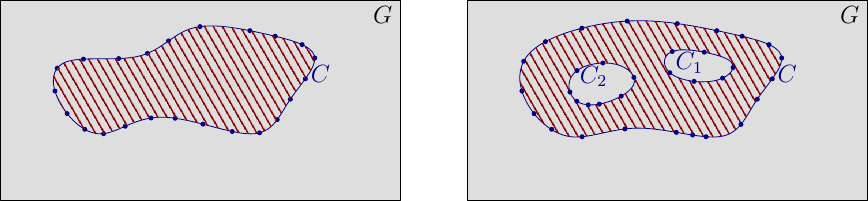}
	\caption{Schematic view of the subgraph $H=\interior(G,C)$ defined 
		by the cycle $C$ (left) and the subgraph $\Holes(G,C,\{C_1,C_2\})$  
		defined by a few cycles (right). They are both marked with stripes.}
	\label{fig:example}
\end{figure}

%%%%%%%%%%%%%%%%%%%%%%%%%%%%%%%%%%%%%%%%%%%%%%%%%%%%%%%%%%%%%%%%%%%%%%%%%%%%%%%%%%%%%%%%%%%%%%%%%%%%%%%%%%%%%%
\subparagraph{Notation and basic concepts.}
Through our discussion we will assume that the host graph $G$ is fixed 
and in the notation we will often drop the dependency on $G$.

For each subgraph $H$ of $G$, the \DEF{boundary} of $H$ (in $G$)
is the set $\partial H$  of vertices 
of $H$ that are incident to some edge outside $H$.
Thus, $\partial H = \{ v\in V(H)\mid \exists uv\in E(G)\setminus E(H)\}$;
see \cref{fig:boundary} for a schema.
Note that each path in $G$ that has edges from $E(H)$ and edges from 
$E(G)\setminus E(H)$
must have some vertex of $\partial H$.

\begin{figure}
	\centering
	\includegraphics[page=3,scale=.85]{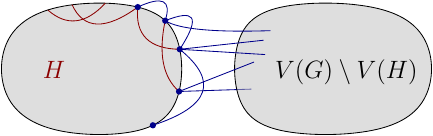}
	\caption{Boundary vertices of $H$ are denoted with blue dots. Note that an edge
		connecting vertices of $H$ may belong to $E(H)$ or not; 
		in the latter case the vertices belong to $\partial H$.}
	\label{fig:boundary}
\end{figure}

For each subgraph $H$ of $G$, let $\Hout$ be the \DEF{complement} of $H$ in $G$,
defined as the subgraph with edges not contained in $H$.
Thus, $\Hout=(V(G),E(G)\setminus E(H))$. The graph $\Hout$ may have isolated vertices.
For distances in $\Hout$ we write $d_{\Hout}(\cdot,\cdot)$.

We will often use without explicit mention that for any subgraph $H$ of $G$ 
and any vertices $x,y\in V(H)$ we have $d_G(x,y)\le d_H(x,y)$
and $d_G(x,y)\le d_{\Hout}(x,y)$.

The \DEF{Wiener index with marked vertices} for a graph $G$ with respect
to $U\subseteq V(G)$ is $\sum_{x,y\in U} d_G(x,y)$.
The classical Wiener index is the case of $U=V(G)$.
For unweighted graphs, this is a very similar to the vertex-weighted 
medians considered by Bandelt and Chepoi~\cite{BandeltC02} 
and it is a particular case of the vertex-weighted Wiener index 
considered by B{\'{e}}n{\'{e}}teau et al.~\cite{BeneteauCCV22} for median graphs.
In their vertex-weighted version, the marked vertices have weight $1$
and the vertices to be ignored in the sum have weight $0$.

We will use that several of the existing algorithms for computing distances
and the Wiener index in graphs work in the so-called \DEF{comparison-addition model}, 
which means that they are manipulating and comparing sums of the edge lengths,
but they do not make other operations with the edge lengths.
Precise references to the algorithms will be given below when considering 
special classes of graphs.

%%%%%%%%%%%%%%%%%%%%%%%%%%%%%%%%%%%%%%%%%%%%%%%%%%%%%%%%%%%%%%%%%%%%%%%%%%%%%%%%%%%%%%%%%%%%%%%%%%%%%%%%%%%%%%%%%%%%%%%%%%%%%%%%%%%%%%%%%%%%%%%%%%%%%%%%%%%%%%%%%%%%%%%%%%%%%%%%%%%%%%%%%%%%%%%%%%%%%%%%%%%%%%%%%%%%%%%%%%%%%%%%%%
\section{Characterizing isometric subgraphs}
\label{sec:isometric_criteria}

In this section we provide different criteria to test whether a subgraph $H$ of $G$ 
is isometric.

\begin{lemma}
\label{le:isometric-boundary}
	Let $G$ be a connected graph and let $H$ be a subgraph of $G$.
	Then $H$ is an isometric subgraph of $G$ if and only if
	\[
		\forall x,y\in \partial H: ~~~ d_{H}(x,y) \le d_{\Hout}(x,y).
	\]
\end{lemma}	
\begin{proof}
	One direction is obvious: if $H$ is an isometric subgraph of $G$,
	then for $\forall x,y\in V(H)$ we have $d_{H}(x,y) = d_G(x,y)\le d_{\Hout}(x,y)$.
	
	To show the other direction, assume that for all distinct 
	$x,y\in \partial H$ we have $d_{H}(x,y) \le d_{\Hout}(x,y)$. 
	Consider any two vertices $x,y$ in $H$ and a shortest path $\pi$ in $G$ between $x$ and $y$.
	If $\pi$ is contained in $H$, then $d_G(x,y)\le d_H(x,y)\le \ell(\pi) =d_G(x,y)$
	and therefore $d_G(x,y)=d_H(x,y)$.
	If $\pi$ is not contained in $H$, then $\pi$ contains some subpath in $\Hout$.
	We split $\pi$ into maximal connected subpaths $\pi_1,\dots,\pi_k$ of $\pi$ such
	that each of them is contained either in $H$ or in $\Hout$.
	For each $\pi_i$ contained in $\Hout$, its endpoints $x_i,y_i$ are in $\partial H$.
	For each index $i\in [k]$ we define
	\[
		\pi'_i= \begin{cases}
				\pi_i &\text{if $\pi_i$ contained in $H$},\\
				\text{shortest path from $x_i$ to $y_i$ in $H$},  &\text{if $\pi_i$ contained in $\Hout$}.
			\end{cases}
	\]
	Note that in the second case such a path $\pi'_i$ in $H$ always exists because 
	by assumption $d_H(x_i,y_i)\le d_{\Hout}(x_i,y_i) = \ell(\pi_i)<+\infty$.
	
	Because of our assumption we have
	\begin{align*}
		\text{if $\pi_i$ inside $H$:}&~~~ \ell(\pi'_i) = \ell(\pi_i);\\
		\text{if $\pi_i$ inside $\Hout$:}&~~~ \ell(\pi'_i) = d_H(x_i,y_i) \le d_{\Hout} (x_i,y_i) = \ell(\pi_i).
	\end{align*}
	Thus $\ell(\pi'_i) \le \ell(\pi_i)$ for each index $i\in [k]$.
	Let $\pi'$ be the $x$-to-$y$ walk obtained by concatenating $\pi'_1,\dots,\pi'_k$. 
	Since $\pi'$ is contained in $H$ because each $\pi'_i$ is contained in $H$,
	we conclude that
	\[
		d_G(x,y) = \ell(\pi) = \sum_{i=1}^k \ell(\pi_i) \ge \sum_{i=1}^k \ell(\pi'_i) = \ell(\pi') 
		\ge d_H(x,y) \ge d_G(x,y),
	\]
	and therefore $d_G(x,y)=d_H(x,y)$.
\end{proof}

The next two characterizations use the Wiener index with marked vertices.

\begin{lemma}
\label{le:isometric-sum}
	Let $G$ be a connected graph and let $H$ be a subgraph of $G$.
	Then $H$ is an isometric subgraph of $G$ if and only if
	\[
		\sum_{x,y\in V(H)} d_G(x,y) = \sum_{x,y\in V(H)} d_H(x,y).
	\]
\end{lemma}
\begin{proof}
	One direction is obvious: if $H$ is an isometric subgraph of $G$, 
	then $d_H(x,y)=d_G(x,y)$ for all $x,y\in V(H)$ and both sums 
	in the lemma have the same value.
	
	To show the other direction, note first that, since $H$ is a 
	subgraph of $G$, for all $x,y\in V(H)$ we have $d_G(x,y)\le d_H(x,y)$.
	If $H$ is not an isometric subgraph of $G$, then
	there exist distinct $x_0,y_0\in V(H)$ such that $d_H(x_0,y_0)\neq d_G(x_0,y_0)$,
	which means that $d_G(x_0,y_0)< d_H(x_0,y_0)$.
	Therefore
	\begin{align*}
		\sum_{x,y\in V(H)} d_G(x,y)~&=~ 
				2\cdot d_G(x_0,y_0) + \sum_{\begin{minipage}{2.2cm}\centering \scriptsize $x,y\in V(H)$\\ $\{x,y\}\neq \{x_0,y_0\}$\end{minipage}} d_G(x,y)\\
				~&<~
			2\cdot d_H(x_0,y_0) + \sum_{\begin{minipage}{2.2cm}\centering \scriptsize $x,y\in V(H)$\\ $\{x,y\}\neq \{x_0,y_0\}$\end{minipage}} d_H(x,y)
				~=~ \sum_{x,y\in V(H)} d_H(x,y). \qedhere
	\end{align*}
\end{proof}

In fact, it suffices to add the distances between vertices on the boundary
of $H$, or any superset of them.

\begin{lemma}
\label{le:isometric-sum-boundary}
	Let $G$ be a connected graph, let $H$ be a subgraph of $G$ and let $U$
	be such that $\partial H \subseteq U\subseteq V(H)$.
	Then $H$ is an isometric subgraph of $G$ if and only if
	\[
		\sum_{x,y\in U} d_H(x,y) = \sum_{x,y\in U} d_G(x,y).
	\]
\end{lemma}
\begin{proof}
	Again, one direction is obvious: if $H$ is an isometric subgraph of $G$, 
	then $d_H(x,y)=d_G(x,y)$ for all $x,y\in V(H)$, and since $U\subseteq V(H)$,
	both sums in the lemma are equal.
	
	To show the other direction, we note first that
	for all $x,y\in U$ we have $d_G(x,y)\le d_H(x,y)$.
	If $H$ is not an isometric subgraph of $G$, then by \cref{le:isometric-boundary}
	there exist $x_0,y_0\in \partial H\subseteq U$ such that 
	$d_H(x_0,y_0)> d_{\Hout}(x_0,y_0) \ge d_G(x_0,y_0)$.
	Then a sum like the one in the proof of \cref{le:isometric-sum} implies
	\[
		\sum_{x,y\in U} d_G(x,y) ~<~ \sum_{x,y\in U} d_H(x,y). \qedhere
	\]
\end{proof}

%%%%%%%%%%%%%%%%%%%%%%%%%%%%%%%%%%%%%%%%%%%%%%%%%%%%%%%%%%%%%%%%%%%%%%%%%%%%%%%%%%%%%%%%%%%%%%%%%%%%%%%%%%%%%%%%%%%%%%%%%%%%%%%%%%%%%%%%%%%%%%%%%%%%%%%%%%%%%%%%%%%%%%%%%%%%%%%%%%%%%%%%%%%%%%%%%%%%%%%%%%%%%%%%%%%%%%%%%%%%%%%%%%
\section{Characterizing convex subgraphs}
\label{sec:convex_criteria}

Here we provide criteria to identify when a subgraph $H$ of $G$ 
is convex. The structure of the statements and the proofs is parallel
to the isometric case considered in \cref{sec:isometric_criteria}.

\begin{lemma}
\label{le:convex-boundary}
	Let $G$ be a connected graph and let $H$ be a subgraph of $G$.
	Then $H$ is a convex subgraph of $G$ if and only if
	\[
		\forall x,y\in \partial H, ~x\neq y: ~~~ d_{H}(x,y) < d_{\Hout}(x,y).
	\]
\end{lemma}	
\begin{proof}
	If $H$ is a convex subgraph of $G$,	then, for all distinct $x,y\in \partial H$,
	all the shortest paths connecting
	$x$ to $y$ are contained in $H$, and thus each path connecting $x$ to $y$ 
	that is contained in the complement of $H$, if any exists,
	must be strictly longer.
	
	To show the other direction, assume that for all distinct 
	$x,y\in \partial H$ we have $d_{H}(x,y) < d_{\Hout}(x,y)$. 
	Assume, for the sake of reaching a contradiction, that $H$ is not a convex
	subgraph of $G$. Thus there exist two vertices $x,y$ in $H$ and a shortest
	path $\pi$ in $G$ between $x$ and $y$ that is {\em not} contained in $H$.
	Since $\pi$ is not contained in $H$, it contains some subpaths in $\Hout$.
	We split $\pi$ into maximal connected subpaths $\pi_1,\dots,\pi_k$ of $\pi$ such
	that each of them is contained either in $H$ or in $\Hout$. 
	For each $\pi_i$ contained in $\Hout$, its endpoints $x_i,y_i$ are in $\partial H$.
	For each index $i\in [k]$ we define
	\[
		\pi'_i= \begin{cases}
				\pi_i, &\text{if $\pi_i$ contained in $H$},\\
				\text{shortest path from $x_i$ to $y_i$ in $H$},  &\text{if $\pi_i$ contained in $\Hout$}.
			\end{cases}
	\]
	Note that, in the second case, such a path $\pi'_i$ in $H$ always exists because 
	by assumption $d_H(x_i,y_i)< d_{\Hout}(x_i,y_i)\le \ell(\pi_i)<+\infty$.

	The hypothesis implies that, for each index $i\in [k]$,
	\begin{align*}
		\text{if $\pi_i$ inside $H$:}&~~~ \ell(\pi'_i) = \ell(\pi_i);\\
		\text{if $\pi_i$ inside $\Hout$:}&~~~ \ell(\pi'_i) = d_H(x_i,y_i) < d_{\Hout} (x_i,y_i) = \ell(\pi_i).
	\end{align*}
	Let $\pi'$ be the $x$-to-$y$ walk obtained by concatenating $\pi'_1,\dots,\pi'_k$. 
	Because for some index $i\in [k]$ we have $\ell(\pi'_i) < \ell(\pi_i)$ and for all indices $i\in [k]$
	we have $\ell(\pi'_i) \le \ell(\pi_i)$, we conclude that
	\[
		d_G(x,y) = \ell(\pi) = \sum_{i=1}^k \ell(\pi_i) > \sum_{i=1}^k \ell(\pi'_i) = \ell(\pi') \ge d_G(x,y),
	\]
	which implies the contradiction $d_G(x,y)>d_G(x,y)$.
\end{proof}

For the following characterizations, we are going to decrease 
slightly the length of some edges of $G$ so that no new shortest 
paths are created. For this, we define
\begin{equation*}\label{eq:eps}
	\eps = \eps(G,\ell) = \frac{ \min \{ \ell(\pi)-\ell(\pi')\mid \text{$\pi,\pi'$ paths in $G$ with $\ell(\pi)>\ell(\pi')$}\}}{n}.
\end{equation*}
First note that we can indeed take the minimum because there is a finite number
of paths in $G$. Moreover, for each edge $e\in E(G)$, we have $\eps\le \ell(e)/n$
because $e$ is a path in $G$, as well as the empty path.
For a subgraph $H$ of $G$ and a $\delta$ with $0<\delta\le \eps$,
we define the edge lengths $\hat\ell$
by shortening the edges {\em not} in $H$ by $\delta$. Thus
\[
	\hat\ell(e)=\begin{cases}
			\ell(e), &\text{if $e\in E(H)$},\\
			\ell(e)-\delta, &\text{if $e\in E(G)\setminus E(H)$}.	
			\end{cases}
\]
Note that $\hat\ell(e)>0$ for all $e\in E(G)$ because $\delta\le \eps < \ell(e)$.
The lengths $\hat\ell$ depend on $H$ and $\delta$, but we drop this dependency 
to avoid cluttering the notation.

\begin{lemma}
\label{le:perturbation}
	If $G$ is connected and $0< \delta \le \eps$, each shortest path in $G$ 
	with respect to $\hat\ell$
	is a shortest path in $G$ with respect to $\ell$. 
	\textup(The converse is not necessarily true.\textup)
\end{lemma}
\begin{proof}
	Let $\pi$ be a shortest $x$-to-$y$ path in $G$ with respect to $\hat\ell$.
	This means that, for any other $x$-to-$y$ path $\pi'$ in $G$
	we have $\hat\ell(\pi)\le \hat\ell(\pi')$.
	We have 
	\begin{align*}
		\ell(\pi) ~&=~ \sum_{e\in E(\pi)} \ell(e) 
				~=~  \sum_{e\in E(\pi)} \hat\ell (e) + |E(\pi) \setminus E(H)|\cdot \delta
				~=~ \hat\ell(\pi) + |E(\pi) \setminus E(H)|\cdot \delta\\
				~&\le~ \hat\ell(\pi') + (n-1) \delta ~<~ \ell(\pi') + n \eps,
	\end{align*}
	where in the last inequality we used that the $\hat\ell$-length is
	never longer than the $\ell$-length.
	Therefore $\ell(\pi) < \ell(\pi') + n\eps$.
	From the definition of $\eps$, we conclude that $\ell(\pi)\le \ell(\pi')$. 
	Since this holds for each $x$-to-$y$ path $\pi'$, the path $\pi$
	is a shortest $x$-to-$y$ path with respect to $\ell$.	
\end{proof}

Let $\hat d_G(\cdot,\cdot)$ denote the distance in $G$ with respect to 
the edge lengths $\hat\ell$. 

\begin{lemma}
\label{le:convex-sum}
	Let $G$ be a connected graph, let $H$ be a subgraph of $G$, 
	let $\delta$ be a real value such that $0< \delta\le \eps$,
	and let $\hat d_G$ be the corresponding perturbed distance.
	Then $H$ is a convex subgraph of $G$ if and only if
	\[
		\sum_{x,y\in V(H)} \hat d_G(x,y) = \sum_{x,y\in V(H)} d_H(x,y).
	\]
\end{lemma}	
\begin{proof}
	Assume first that $H$ is a convex subgraph of $G$. 
	Consider any two vertices $x,y$ of $H$ and let $\pi_{x,y}$
	be a shortest path in $G$ from $x$ to $y$ with respect to $\hat\ell$.
	Therefore $\hat d_G(x,y) = \hat\ell(\pi_{x,y})$.
	Because of \cref{le:perturbation}, the path $\pi_{x,y}$ is a shortest
	path in $G$ with respect to $\ell$, which means that 
	$\ell(\pi_{x,y}) = d_G(x,y)$.
	Because $H$ is convex, the path $\pi_{x,y}$ is contained in $H$,
	and since the edges of $H$ have the same
	length in $\ell$ and in $\hat\ell$, we have
	$\hat\ell(\pi_{x,y}) = \ell(\pi_{x,y})$.
	Putting the preceding equalities together and using that $H$ is a convex 
	(and thus isometric) subgraph of $G$, we obtain
	\[
		\hat d_G(x,y) = \hat\ell(\pi_{x,y}) =
		\ell(\pi_{x,y}) = d_G(x,y) = d_H(x,y).
	\]
	Since this equality holds for each $x,y\in V(H)$, we obtain
	\[
		\sum_{x,y\in V(H)} \hat d_G(x,y) = \sum_{x,y\in V(H)} d_H(x,y).
	\]
	
	To show the other direction, note first that 
	$\hat d_G(x,y)\le d_G(x,y) \le d_H(x,y)$
	for all $x,y\in V(H)$ because $\hat\ell(e) \le \ell(e)$
	for all edges $e\in E(G)$ (for the first inequality) 
	and $H$ is a subgraph of $G$ (for the second inequality).
	Assume now that $H$ is not a convex subgraph of $G$. This means that
	there exist vertices $x_0,y_0\in V(H)$ and a $x_0$-to-$y_0$
	shortest path $\pi_0$ that is not contained in $H$. Since $\pi_0$
	has some edge in $E(G)\setminus E(H)$, we have
	$\hat\ell(\pi_0) < \ell(\pi_0)$ and therefore
	$\hat d_G(x_0,y_0)< d_G(x_0,y_0) \le d_H(x_0,y_0)$.
	Therefore
	\begin{align*}
		\sum_{x,y\in V(H)} \hat d_G(x,y) ~&=~ 
			2\cdot \hat d_G(x_0,y_0) + \sum_{\begin{minipage}{2.2cm}\centering \scriptsize $x,y\in V(H)$\\ $\{x,y\}\neq \{x_0,y_0\}$\end{minipage}} \hat d_G(x,y)\\
		~&<~ 2\cdot d_H(x_0,y_0) + \sum_{\begin{minipage}{2.2cm}\centering \scriptsize $x,y\in V(H)$\\ $\{x,y\}\neq \{x_0,y_0\}$\end{minipage}} d_H(x,y) 
		~=~ \sum_{x,y\in V(H)} d_H(x,y). \qedhere
	\end{align*}
\end{proof}

\begin{lemma}
\label{le:convex-sum-boundary}
	Let $G$ be a connected graph, let $H$ be a subgraph of $G$, 
	let $U$ be such that $\partial H \subseteq U\subseteq V(H)$,
	let $\delta$ be a real value such that $0< \delta\le \eps$,
	and let $\hat d_G$ be the corresponding perturbed distance.
	Then $H$ is a convex subgraph of $G$ if and only if
	\[
		\sum_{x,y\in U} d_H(x,y) = \sum_{x,y\in U} \hat d_G(x,y).
	\]
\end{lemma}	
\begin{proof}
	In the first part of the proof of \cref{le:convex-sum} we have seen
	that, if $H$ is a convex subgraph of $G$, then for all $x,y\in V(H)$
	we have $\hat d_G(x,y)=d_H(x,y)$. The equality of the two sums follows.
	
	For the other direction, we start assuming that $H$ is not a convex
	subgraph of $G$. If $H$ is not connected, then the left sum
	is $\infty$ while the right side is bounded, and therefore they
	are distinct. We continue assuming that $H$ is connected.
	By \cref{le:convex-boundary}, there exist distinct vertices $x_0,y_0\in \partial H$
	such that $d_H(x_0,y_0)\ge d_{\Hout}(x_0,y_0)$.
	Let $\pi_0$ be a shortest $x_0$-to-$y_0$ path in $H$
	and let $\pi'_0$ be a shortest $x_0$-to-$y_0$ path in $\Hout$.
	Note that $\pi'_0$ indeed exists because $d_{\Hout}(x_0,y_0)$ is bounded.
	We then have $\ell(\pi_0) = d_H(x_0,y_0) \ge d_{\Hout}(x_0,y_0) = \ell(\pi'_0)$.
	Since $\pi'_0$ has edges outside $E(H)$, we have $\hat \ell(\pi'_0) < \ell(\pi'_0)$,
	and therefore 
	\[
		\hat d_G(x_0,y_0) \le \hat \ell(\pi'_0) < \ell(\pi'_0) \le \ell(\pi_0) = d_H(x_0,y_0). 
	\]
	In summary, $\hat d_G(x_0,y_0)< d_H(x_0,y_0)$ for some $x_0,y_0\in \partial H$.
	In general, for all $x,y\in V(H)$ we have 
	$\hat d_G(x,y)\le d_G(x,y) \le d_H(x,y)$
	because $\hat\ell(e) \le \ell(e)$
	for all edges $e\in E(G)$ (for the first inequality) 
	and $H$ is a subgraph of $G$ (for the second inequality).
	From this it follows that for $U$ such that $\partial H \subseteq U\subseteq V(H)$
	we have
	\[
		\sum_{x,y\in U} \hat d_G(x,y) < \sum_{x,y\in U} d_H(x,y) .\qedhere
	\]
\end{proof}

%%%%%%%%%%%%%%%%%%%%%%%%%%%%%%%%%%%%%%%%%%%%%%%%%%%%%%%%%%%%%%%%%%%%%%%%%%%%%%%%%%%%%%%%%%%%%%%%%%%%%%%%%%%%%%%%%%%%%%%%%%%%%%%%%%%%%%%%%%%%%%%%%%%%%%%%%%%%%%%%%%%%%%%%%
\section{Conditional lower bound}
\label{sec:lower}

Our conditional lower bound is based on the following assumptions,
where $\langle \cdot, \cdot \rangle$ is the inner product
and the diameter of a graph $G$ is $\diam(G)=\max_{x,y\in V(G)} d_G(x,y)$.

\begin{OVC}[OVC,~\cite{Williams05}]
For every $\eps>0$ there is a $c\ge 1$ such that the following problem
cannot be solved in $O(n^{2-\eps})$ time: given $n$ binary vectors
$v_1,\dots,v_n\in \{0,1\}^d$, where $d=\lceil c\log n \rceil$, 
are there $i,j\in [n]$ such $\langle v_i, v_j\rangle = 0$?
\end{OVC}

\begin{SETH}[SETH,~\cite{ImpagliazzoP01,ImpagliazzoPZ01}]
For every $\eps>0$ there exists an integer $k\ge 3$ such that
$k$-SAT with $n$ variables cannot be solved in $O(2^{(1-\eps)n})$ time.
\end{SETH}

Williams~\cite{Williams05} posed OVC and showed that it is implied
by SETH. One of the interesting consequences of OVC, and therefore SETH, 
is the following result.

\begin{theorem}[Consequence of Theorem 9 in Roditty and Vassilevska Williams~\cite{RodittyW13}]
\label{th:diameter}
	If for some constant $\eps>0$ one can distinguish in $O(n^{2-\eps})$ time 
	between diameter $2$ and $3$ in undirected graphs with $n$ vertices 
	and $O(n)$ edges, all of unit length, 
	then OVC and SETH fail.
\end{theorem}

\begin{theorem}
\label{th:reduction}
	Assume that, for some $\eps>0$,	there is an algorithm that in  
	$O(n^{2-\eps})$ time solves the following problem:
	given an undirected graph $G$ with $n$ vertices and $O(n)$ edges, 
	all of unit length, and given a subgraph $H$
	of $G$, decide whether $H$ is an isometric subgraph of $G$.
	Then OVC and SETH fail.
	
	The same result holds for deciding if $H$ is a convex subgraph of $G$.
\end{theorem}
\begin{proof}
	We will reduce from the diameter problem considered in \cref{th:diameter}.
	Let $H$ be a graph with $O(n)$ vertices and $O(n)$ edges, all edges of
	unit length, for
	which we want to distinguish whether it has diameter $2$ or $3$. 
	Consider the graph $G=G(H)$ defined as follows; see \cref{fig:reduction}
	for an example.
	We first construct the graph $H_3$ from $H$ by subdividing each edge
	of $H$ twice.
	Then, we add a new vertex $v_*$ and connect it to each vertex of $V(H)$ 
	with a $4$-edge path using new interior vertices. Let $G=G(H)$ be the resulting
	graph.
	Note that $H_3$ is a subgraph of $G$ and $\partial H_3 = V(H)$.

	\begin{figure}
	\centering
		\includegraphics[page=1,width=\textwidth]{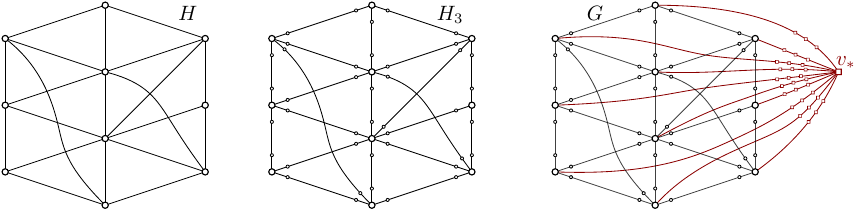}
		\caption{Example for the construction of \cref{th:reduction}.
			Left: a graph $H$. Center: the graph $H_3$. 
			Right: the graph $G$ with $H_3$ as a subgraph.}
		\label{fig:reduction}
	\end{figure}

	We assign unit length to each edge of $G$.
	For all $x,y\in V(H)$, we have $d_{H_3}(x,y)= 3\cdot d_H(x,y)$
	and $d_G(x,y)= \min \{ 3\cdot d_H(x,y), 8 \}$,
	where the $8$ comes from the path in $G$ through $v_*$.

	We now show that $H_3$ is an isometric subgraph of $G$ if and only if $H$
	has diameter $2$.
	If $H$ has diameter $2$, then for all $x,y\in V(H)$ we have
	\[
		d_G(x,y)= \min \{ 3\cdot d_H(x,y), 8 \} = 3\cdot d_H(x,y) = d_{H_3}(x,y).
	\]
	Since $\partial H_3=V(H)$,
	we have $\sum_{x,y\in \partial H_3} d_G(x,y)= \sum_{x,y\in \partial H_3} d_{H_3}(x,y)$
	and $H_3$ is an isometric subgraph of $G$ because of \cref{le:isometric-sum-boundary}.
	If $H$ has diameter $3$, then there is 
	some $x_0,y_0\in V(H)$ with $d_H(x_0,y_0)=3$,
	and therefore $d_{H_3}(x_0,y_0) = 3\cdot d_H(x_0,y_0)= 9 > 8 \ge d_G(x_0,y_0)$.
	It follows that $H_3$ is not an isometric subgraph of $G$.
	
	We now show that $H_3$ is a convex subgraph of $G$ if and only if $H$
	has diameter $2$. 
	Let $\hat d$ denote the distance in $G$ when we decrease the length of the edges
	of $E(G)\setminus E(H_3)$ by a sufficiently small value $\delta>0$. 
	For concreteness, we can take $\delta < 1/10$ in our discussion. 
	Then, for all $x,y\in V(H)$, we have $d_{H_3}(x,y)= 3\cdot d_H(x,y)$
	and $\hat d_G(x,y)= \min \{ 3\cdot d_H(x,y), 8-8\delta \}$.
	
	If $H$ has diameter $2$, then, for all $x,y\in V(H)$, we have
	\[
		\hat d_G(x,y)= \min \{ 3\cdot d_H(x,y), 8-8\delta \} = 3\cdot d_H(x,y) = d_{H_3}(x,y),
	\]
	where we have used that $3\cdot d_H(x,y)\le 6$ and $8\delta <1$.
	Since $\partial H_3=V(H)$,
	we have $\sum_{x,y\in \partial H_3} \hat d_G(x,y)= \sum_{x,y\in \partial H_3} d_{H_3}(x,y)$,
	and $H_3$ is a convex subgraph of $G$ because of \cref{le:convex-sum-boundary}.

	If $H$ has diameter $3$, then we have already seen that $H_3$ is not 
	an isometric subgraph of $G$, and thus cannot be a convex subgraph.
	
	If $H$ has $O(n)$ vertices and $O(n)$ edges, then 
	$G$ has $|V(H)|+2\cdot|E(H)|+3\cdot|V(H)|+1=O(n)$ vertices 
	and $3\cdot|E(H)|+4\cdot|V(H)|=O(n)$ edges. The construction
	of $G$ and $H_3$ from $H$ takes linear time.
	It follows that, if we could decide whether $H_3$ is isometric (or convex) 
	in $O(n^{2-\eps})$ time, for some $\eps>0$, we could decide whether $H$
	has diameter $2$ or $3$ in $O(n^{2-\eps})$ time, 
	and therefore OVC and SETH would be false by \cref{th:diameter}.
\end{proof}

Note that in our reduction the host graph has diameter $O(1)$.
It seems that adapting the original proof of \cref{th:diameter} by
\cite{RodittyW13} one can obtain host graphs with even smaller diameter.
We did not pursue the details of such a construction.

%%%%%%%%%%%%%%%%%%%%%%%%%%%%%%%%%%%%%%%%%%%%%%%%%%%%%%%%%%%%%%%%%%%%%%%%%%%%%%%%%%%%%%%%%%%%%%%%%%%%%%%%%%%%%%%%%%%%%%%%%%%%%%%%%%%%%%%%%%%%%%%%%%%%%%%%%%%%%%%%%%%%%%%%%
\section{Subquadratic algorithms for special classes of graphs}
\label{sec:viaWiener}

In this section we provide algorithms to test whether a 
subgraph is isometric or convex, when the host
graph is planar or has small treewidth.
The idea is to use the criteria of \cref{le:isometric-sum,le:convex-sum}
together with modifications of previous algorithms to compute the classical
Wiener index.

\subsection{Shortening some edges a bit}
\label{sec:shortening}
To test whether $H\subseteq G$ is convex using \cref{le:convex-sum}
we have to shorten the edges of $E(G)\setminus E(H)$ by a small
enough amount $\delta$, where $0<\delta\le \eps(G,\ell)$. 
We do not know how to compute $\eps(G,\ell)$ efficiently, 
or some other value that could replace it. We solve this by 
perturbing the edge lengths infinitesimally, as follows.
(For unit-length edges this is easy by taking $\delta=1/2n$.)

We define a new, composite length $\ell_2$ for $G$ using $2$-tuples.
The intuition is that a pair $(a,b)\in \RR_{\ge 0}\times \ZZ$ 
represents the real value $a+b\cdot \delta$ for an infinitesimal $\delta>0$.
This means that we can compare two lengths $(a,b)$ and $(a',b')$ using
the lexicographic comparison: $(a,b)$ is smaller than $(a',b')$, denoted
by $(a,b)\prec (a',b')$,
if and only if $a<a'$ or if $a=a'$ and $b<b'$. 
This lexicographic comparison is compatible with the test of
whether $a+b\cdot \delta < a'+b'\cdot \delta$ for an infinitesimal $\delta>0$.

To carry the perturbation, we define a new edge length $\ell_2$ in the following way:
\[
	\forall e\in E(G):~~~~ \ell_2(e)= \begin{cases}
			(\ell(e),0) &\text{if $e\in E(H)$,}\\
			(\ell(e),-1) &\text{if $e\in E(G)\setminus E(H)$.}
		\end{cases}
\]
The length of a walk (or path) $\pi$ is defined as the vector-sum of the 
the $2$-tuples $\ell_2(e)$ over $e\in E(\pi)$, with multiplicity.
Therefore, we have
\[
	\ell_2(\pi) = \big( \ell(\pi),~ -|E(\pi)\setminus E(H)| \big),
\]
which represents the value $\ell(\pi)-\delta\cdot |E(\pi)\setminus E(H)|$
for an infinitesimal $\delta>0$ and where $E(\pi)$ is interpreted
as a multiset.

Classical algorithms for computing shortest paths work in the
comparison-addition model and thus can compute distances
using this new $2$-tuple distances without
an asymptotic increase in the running time.
(For the sake of comparison, it is not clear how to adapt efficiently 
an algorithm that would multiply edge lengths because 
powers of $\delta$ would start showing up and a constant-size tuple would
not suffice anymore.)

\subsection{Graphs of small treewidth}
Let $G$ be a graph, possibly with edge lengths,
and let $H$ be a subgraph of $G$. Let $\tw(G)$ denote the treewidth
of $G$. It is well-known that $H$ has treewidth at most $\tw(G)$.

Cabello and Knauer~\cite{CabelloK09} showed that, for each
constant $k\ge 3$, the Wiener index of graphs of treewidth at most $k$ 
can be computed in $O(n \log^{k-1} n)$ time;
the constant hidden in the $O$-notation depends on $k$.
Bringmann, Husfeldt and Magnusson~\cite{BringmannHM20} carried out the analysis
making the dependency on the treewidth explicit and showed that
the Wiener index can be computed in
$n^{1+\eps} 2^{O(\tw(G))}$, for each constant $\eps>0$;
the constant hidden in the $O$-notation depends on $\eps$.

Both papers use data structures for orthogonal range searching. The vertices of $G$
are used to define a point set and to define ranges.
It is easy to see that those results can be adapted to compute the sum
$\sum_{x,y\in U} d_G(x,y)$ for any given $U\subseteq V(G)$.
Indeed, it suffices to define the points and the ranges only for the vertices of $U$.
This is explicitly mentioned in~\cite[Section 5]{Cabello22}.
The running times of the results are not affected by this. (Actually, 
if $\log |U|= o(\log n)$, the asymptotic running times can be
slightly improved, but we omit this in our worst-case analysis.)

Next, we note that the algorithms in \cite{BringmannHM20, CabelloK09} 
work in the comparison-addition model.
This means that they can handle edge lengths defined by a $2$-tuple,
as discussed in \cref{sec:shortening}.
We conclude that those algorithms can be used, without changing the asymptotic
running time, to compute
\[
	\sum_{x,y\in V(H)}d_G(x,y),~~~	\sum_{x,y\in V(H)}\hat d_G(x,y),
		~~~\sum_{x,y\in V(H)}d_H(x,y),
\]
where $\hat d_G(x,y)$ is defined via $2$-tuples to model a shortening of the edges
of $E(G)\setminus E(H)$ by an infinitesimal $\delta>0$.
We can then compare them, taking into account that the $\ell_2$-length
$(a,b)$ is equal to the $\ell$-length $a$ if and only if $b=0$.
From Lemmas~\ref{le:isometric-sum} and~\ref{le:convex-sum}, the following is now immediate.

\begin{theorem}
	Let $k\ge 3$ be a constant. Given a graph $G$ with $n$ vertices and treewidth at most $k$, 
	we can test whether a given subgraph of $G$ is convex or isometric in $O(n \log^{k-1} n)$ time.
	The constant hidden in the $O$-notation depends on $k$.

	We can decide whether a given subgraph of $G$ is convex or isometric
	in $n^{1+\eps} 2^{O(\tw(G))}$ time, for each constant $\eps>0$.
	The constant hidden in the $O$-notation depends on $\eps$.
\end{theorem}

\subsection{Planar graphs}

Cabello~\cite{Cabello19} gave a randomized algorithm to compute the Wiener index 
of a planar graph $G$ in $\tilde O(n^{11/6})$ expected time.
This was improved by Gawrychowski et al.~\cite{GawrychowskiKMS21},
who provided a deterministic algorithm taking $\tilde O(n^{5/3})$ time.
Both algorithms have the property that they can compute
$\sum_{x,y\in U} d_G(x,y)$ for any given $U\subseteq V(G)$.
Also, both algorithms work in the comparison-addition model. 
Therefore, they can also be used to obtain
\[
	\sum_{x,y\in V(H)}d_G(x,y),~~~	\sum_{x,y\in V(H)}\hat d_G(x,y),
		~~~\sum_{x,y\in V(H)}d_H(x,y),
\]
where $\hat d_G(x,y)$ is defined via $2$-tuples to model a shortening
by an infinitesimal $\delta>0$.
This shows the following.

\begin{theorem}
	Given a planar graph $G$ with $n$ vertices, we can test whether a 
	given subgraph of $G$ is convex or isometric in $\tilde O(n^{5/3})$ time.
\end{theorem}

\subsection{Discussion}
We next present classes of graphs illustrating two limitations of
our technique based on the Wiener index with marked vertices.
One limitation is that for testing convexity we have to consider edge-weights, 
even if the original graph is unweighted; the other limitation is that subgraphs 
may lose some useful properties enojoyed by the host graph.

Let us consider the class of $K_h$-minor-free graphs,
where $K_h$ denotes the complete graph on $h$ vertices.
Ducoffe et al.~\cite{DucoffeHV22} showed how to compute the diameter of such graphs 
in subquadratic time for the unweighted version, but left open the 
question of computing the Wiener index in subquadratic time. 
Recent works by Le, Wulff{-}Nilsen, Karczmarz and Zheng~\cite{KarczmarzZ25,LeW24}
have improved the dependency on $h$ and their new techniques 
also work for the Wiener index, but still only for the case of unweighted graphs.
The currently best result~\cite[Theorem~1.3]{KarczmarzZ25}
computes the Wiener index in \emph{unweighted} graphs 
without a $K_h$-minor in $\tilde O(n^{2-1/(3h-2)})$ time. 
The technique can be easily adapted to compute the Wiener index with marked vertices.
Since a subraph of a $K_h$-minor-free graph is obviously also
$K_h$-minor-free, we can use this algorithm to test whether a given subgraph 
is isometric via Lemma~\ref{le:isometric-sum}.
We summarize. 

\begin{theorem}
	Given an \emph{unweighted} $K_h$-minor-free graph $G$ with $n$ vertices, 
	we can test whether a given subgraph of $G$ is isometric 
	in $\tilde O(n^{2-1/(3h-2)})$ time.
\end{theorem}

Since currently we do not know how to compute in subquadratic time
the Wiener index of $K_h$-minor-free graphs with edge weights, 
we cannot make the modification of edge-lengths needed to test 
whether a subgraph is convex via Lemma~\ref{le:convex-sum}. 
Therefore, we cannot make an analogous claim regarding testing convexity, 
even for unweighted graphs. 

\medskip

Let us turn our attention now to unweighted median graphs.
B{\'{e}}n{\'{e}}teau et al.~\cite{BeneteauCCV22} show that the Wiener index
of median graphs can be computed in linear time. In fact, they solve a 
vertex-weighted variant that generalizes the Wiener index 
with marked vertices. 
Consider the problem of testing whether a given subgraph $H$ of a median
graph $G$ is isometric. We can apply their algorithm to compute $\sum_{x,y\in V(H)} d_G(x,y)$
in linear time. However, it is not clear (at least to the author) how to compute
in subquadratic time $\sum_{x,y\in V(H)} d_H(x,y)$, even if we assume that $H$ 
is an isometric subgraph of $G$, and thus a partial cube. When we consider a subgraph
of a median graph, we lose many of its structural properties, even if the subgraph is isometric.
Similarly, for testing convexity using the Wiener index, we would need to 
consider the perturbed weights, which in general does not go well with 
many of the structural properties of median graphs.
(In contrast, as mentioned in the introduction, testing whether a subgraph
of a median graph is convex amounts to testing whether it is gated, a property
that can be tested in linear time.)

%%%%%%%%%%%%%%%%%%%%%%%%%%%%%%%%%%%%%%%%%%%%%%%%%%%%%%%%%%%%%%%%%%%%%%%%%%%%%%%%%%%%%%%%%%%%%%%%%%%%%%%%%%%%%%%%%%%%%%%%%%%%%%%%%%%%%%%%%%%%%%%%%%%%%%%%%%%%%%%%%%%%%%%%%
\section{Subgraphs of plane graphs defined by cycles}
\label{sec:plane}

Let $G$ be a plane graph, that is, a planar graph with a fixed embedding
in the plane. For each cycle $C$ of $G$, let $\gamma_C$ be the Jordan curve defined by $C$.
Such a cycle $C$ defines naturally the \DEF{interior graph}, 
denoted by $\interior(G,C)$,
which is the subgraph of $G$ contained in the closure of the bounded 
set of $\RR^2\setminus \gamma_C$.
Thus, $C$ is contained in $\interior(G,C)$.
In this section we consider the problem of testing whether 
$H=\interior(G,C)$ is a convex or isometric subgraph of $G$.
In $H$ we assume the embedding inherited from $G$ and 
then $C$ is a facial walk of $H$.

\subsection{Data structure for sequences of numbers}

We will use a data structure to store a sequence $a_1,\dots,a_k$ of $k$ values
in $\RR$, where each value is initialized arbitrarily, 
and supporting the following operations:
\begin{itemize}
\item $\DSinit(a_1^0,\dots,a_k^0)$: initializes the data structure, setting $a_i=a_i^0$ for each $i\in [k]$.
\item $\DSmin()$ returns $\min\{ a_1,\dots, a_k\}$.
\item $\DSadd(i,j,\Delta)$ adds the value $\Delta\in \RR$ 
  to $a_i,\dots,a_j$, where $i\le j$.
\end{itemize}

An efficient approach is to use a {\em static} balanced binary search tree storing $k$
elements with keys $[k]=\{1,\dots,k\}$, such that the node with key $i$ stores $a_i$.
The tree is then augmented with data at the nodes to represent
values that have to be added to the whole subtree. 
See for example Choi, Cabello and Ahn~\cite{ChoiCA21} for a comprehensive description.
Another alternative, which is a bit overkilling, is to use dynamic trees, 
such as cut-link trees~\cite{SleatorT83} or top trees~\cite{AlstrupHLT05}; 
in this case we maintain a path on $k$ nodes, 
where the $i$th node stores the value $a_i$. These data structures allow to add
values and to query for the minimum value in a subpath, which in our case corresponds
to a contiguous subsequence. We summarize for later reference.

\begin{lemma}
  \label{le:datastructure-numbers}
  There is a data structure to maintain a sequence
  $a_1,\dots,a_k$ of $k$ numbers that supports the 
  operations $\DSmin$, and $\DSadd$ in  
  $O(\log k)$ time per operation.  The initialization of the data
  structure, $\DSinit$, takes $O(k)$ time.
\end{lemma}

\subsection{Encoding distances within a face}

Let $C$ be a facial cycle in a plane graph $G$ with $n$ vertices,
and let $x_1,\dots,x_k$ be the vertices as we walk along $C$, 
starting from an arbitrary vertex $x_1$ of $V(C)$.
For each $i\in [k]$, let $D_i$ be the sequence of numbers telling
the distances from $x_i$ to all other vertices; thus
$D_i = \big( d_G(x_i,x_j)\big)_{j\in [k]}$.

In a seminal work, Klein~\cite{Klein05} showed 
that in $O(n \log n)$ time one can 
maintain a shortest-path tree of $G$ rooted at a vertex of $C$ 
as the root moves along $C$. 
An alternative point of view based on parametric shortest-path trees
was given by Cabello, Chambers and Erickson~\cite{CabelloCE13},
and a new approach based on divide-and-conquer
has been given by Das et al.~\cite{DasKGW22}.

The key insight is that when the root of the shortest
path tree moves along the boundary of a face, 
each directed arc of $G$ appears in the shortest path tree 
along a contiguous subsequence of roots along the facial walk.
In other words, with the move of the root through the facial
walk, each directed edge enters and exits the shortest path tree
at most once. 

These results are the basis for the following property.

\begin{lemma}
\label{le:encoding}
	In $O(n \log n)$ time we can compute
	sequences $\Lambda_2,\dots,\Lambda_k$ of
	operations with the following properties
	\begin{itemize}
		\item for $i=2,\dots , k$, each $\Lambda_i$ is 
			a sequence of $|\Lambda_i|$ operations of the type
			$\DSadd$ in a sequence $a_1,\dots,a_k$ of numbers;
		\item for $i=2,\dots , k$, the sequence $D_i$ of distances
			from $x_i$ is obtained 
			from the sequence $D_{i-1}$ performing the operations
			given in $\Lambda_i$;
		\item the sequences $\Lambda_i$ together have a linear number
			of operations; that is $\sum_i |\Lambda_i| = O(n)$.
	\end{itemize}
\end{lemma}
\begin{proof}
	Let $T$ be an arbitrary spanning tree of $G$ rooted at 
	a vertex $r$, let $e$ be an edge of $T$, and let $T(r,e)$
	be the subtree of $T-e$ that does not contain the root $r$.
	Because $G$ is plane and $C$ defines a face of $G$, the vertices of 
	$V(C)$ that appear in $T(r,e)$ form a continuous
	subsequence of $C$. See \cref{fig:continuous}.
	Note that the subsequence may be empty, meaning that
	$V(C)\cap V(T(r,e))$ is empty.
	Let $s(T,r,e)$ and $t(T,r,e)$ be the start and the end of the subsequence,
	if it is non-empty, such that
	the vertices $\{ x_{s(T,r,e)},x_{s(T,r,e)+1}, \dots,x_{t(T,r,e)}\}$ (indices modulo $k$)
	are $V(C)\cap V(T(r,e))$.
	
	\begin{figure}
	\centering
		\includegraphics[page=4,width=.9\textwidth]{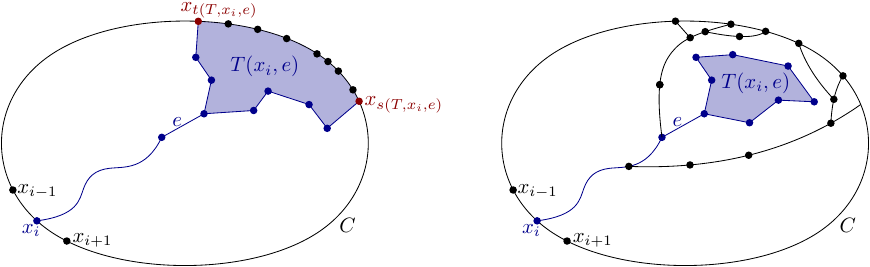}
		\caption{Schematic view of the interaction of $T(x_i,e)$ and the face defined
			by $C$. On the left, $V(T(x_i,e))\cap V(C)$ is not empty, but
			on the right, $V(T(x_i,e))\cap V(C)$ is empty.}
		\label{fig:continuous}
	\end{figure}

	Using data structures for dynamic forests, one 
	can maintain the tree $T$ under swapping of edges (removing an edge and inserting
	an edge while maintaining a tree) and changing the root such that: each update takes
	$O(\log n)$ time; the indices $s(T,x_i,e)$ and $t(T,x_i,e)$ can be obtained 
	in $O(\log n)$ time for any query $(x_i,e)\in V(C)\times E(T)$. 
	The basic idea for this is to maintain the spanning tree $T^*$ 
	of the dual graph that uses the edges dual to those in $E(G)\setminus E(T)$. 
	The two edges defining $s(T,x_i,e)$ and $t(T,x_i,e)$ are, in the dual tree,
	the only edges dual to the edges of $E(C)$ that lie on the unique path of $T^*$ 
	that connects the two faces bounded by $e$. It is easy to extend 
	data structures for dynamic trees~\cite{AlstrupHLT05,SleatorT83}
	to handle such type of queries. The edges dual to $E(C)$ 
	are marked as special in $T^*$ and we have to find the special
	edges that lie in a path of $T^*$ between two given nodes.
	See \cite{CabelloCE13,Klein05} for data structures with similar properties.

	Let $T$ be a tree, let $T'$ be the tree obtained from $T$
	by deleting the edge $e$ and inserting the edge $e'$,
	let $r$ be a vertex of $T$ and $T'$, and
	assume that $e$ in $T$ and $e'$ in $T'$ have a common target; 
	this means that in $T$ we have $e=\dart{x}{z}$
	and in $T'$ we have $e'=\dart{y}{z}$, where the vertex $z$ is the same.	
	See \cref{fig:swap}.
	Let $A=(a_j=d_T(r,x_j))_{j\in [k]}$ the sequence of 
	distances in $T$ from $r$ to $x_j$ (for $j\in [k])$,
	and let $A'=(a'_j=d_{T'}(r,x_j))_{j\in [k]}$ the sequence 
	of distances in $T'$ from $r$ to $x_j$ (for $j\in [k]$).
	Then, we can go from $A$ to $A'$ by performing $O(1)$ operations
	$\DSadd$. More precisely, for each vertex $x_j$ that lies 
	in $T(r,e)$ we have to decrease the value $a_j$ by $d_T(r,z)$
	and increase it by $d_{T'}(r,z)$. (Because $e$ and $e'$ have
	the same target we have $V(T(r,e))=V(T'(r,e'))$.) The values 
	$d_T(r,z)$ and $d_{T'}(r,z)$ can be recovered in $O(\log n)$ time
	from the dynamic data structures for storing $T$ (and thus $T'$).
	Each one of these operations is done with at most two operations
	$\DSadd$. For example, if $s(T,r,e)\le t(T,r,e)$, then we perform
	$\DSadd(s(T,r,e),t(T,r,e),-d_T(r,z))$, but if $s(T,r,e)> t(T,r,e)$, 
	we split it into the two updates
	$\DSadd(s(T,r,e),k,-d_T(r,z))$ and $\DSadd(1,t(T,r,e),-d_T(r,z))$.

	\begin{figure}
	\centering
		\includegraphics[page=5]{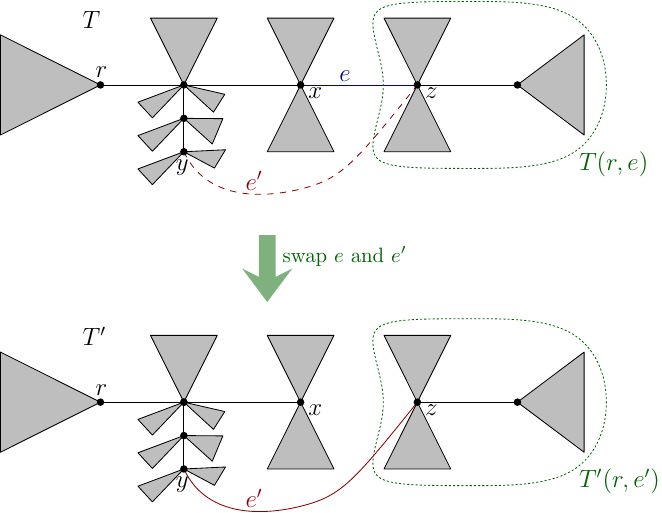}
		\caption{Swapping the edge $e$ and $e'$ to go from $T$
			to $T'$, when $e$ and $e'$ have the same target vertex.}
		\label{fig:swap}
	\end{figure}

	We note that \cite{CabelloCE13,Klein05}
	show how to maintain a shortest-path tree as the root moves along
	the boundary of the face defined by $C$ 
	performing $O(n)$ swaps of the type $T'=(T-e)+e'$ that
	have the property that $e$ and $e'$ have the same target vertex.
	Let $T_{i+1}$ be the resulting tree after we make the swaps
	to transform the shortest-path tree rooted at $x_i$ into the 
	shortest-path tree rooted at $x_{i+1}$.
	The tree $T_{i+1}$ is a shortest path tree from $x_{i+1}$, 
	but currently, the operations we performed in the sequence of numbers 
	encodes the distances in $T_{i+1}$ from $x_i$ to the vertices of $C$.	
	That is, we still have to transform the current sequence 
	$\big( d_{T_{i+1}}(x_i,x_j)\big)_{j\in[k]}$ into
	$\big( d_{T_{i+1}}(x_{i+1},x_j)\big)_{j\in[k]}$.
	We next explain how to do this.
	
	Let $y_1=x_i,\dots,y_b=x_{i+1}$ be the path in $T_{i+1}$ 
	from $x_i$ to $x_{i+1}$. The path may consist of a single edge 
	(if $x_ix_{i+1}\in E(T_{i+1})$) or multiple edges.
	We slide the root of $T_{i+1}$ from $y_1=x_i$ to $y_b=x_{i+1}$
	along this path and update the sequence of numbers to maintain
	the distance from the current root.
	That is, we maintain $\big( d_{T_{i+1}}(y_a,x_j)\big)_{j\in[k]}$
	for $a=1,\dots,b$.
	To this end, for $a=1,\dots b-1$, we add $\ell(y_a y_{a+1})$ to
	all the values of $V(C)\cap V(T_{i+1}(x_{i+1},y_a y_{a+1})$ 
	and subtract $\ell(y_a y_{a+1}))$ to all the values of 
	$V(C)\cap V(T_{i+1}(x_i,y_a y_{a+1}))$. Note that in this
	process the tree $T_{i+1}$ does not change, only the sequence
	of values is updated. The number of updates to the sequence is 
	$O(b)$ and it takes $O(b\log n)$ to compute the updates.

	Finally, as mentioned earlier, when the root of the shortest
	path tree moves along the facial cycle $C$, each directed edge 
	enters and exists at most once into the directed shortest-path tree, 
	if seen as oriented from the root~\cite{CabelloCE13,Klein05}.
	It follows, that the operation in the previous paragraph
	is performed at most twice for each edge of $G$ (once in each direction)
	when the root moves along the boundary of the face defined by $C$.
	The result follows.	
\end{proof}

\subsection{Testing}
We discuss now how to test if $H=\interior(G,C)$ is an isometric
subgraph of $G$.
Recall that $E(\Hout)=E(G)\setminus E(H)$. We add to $\Hout$
the edges of $C$ with very large weights, say $2 \sum_{e\in E(G)}\ell(e)$,
so that the new edges never appear in a shortest path.

Let $x_1,\dots,x_k$ be the vertices of $C$ ordered as they appear in $C$.
The cycle $C$ defines a face in $H$ and in $\Hout$
and therefore we can use \cref{le:encoding} for each of those graphs.
For $H$ we obtain the sequence of operations
$\Lambda_2,\dots,\Lambda_k$ that are needed to maintain
\[
	D_i = \Big( ~a_j=d_H(x_i,x_j) ~\Big)_{j\in [k]},
\]
while for $\Hout$ we obtain the sequence of operations
$\Lambda'_2,\dots,\Lambda'_k$ that are needed to maintain
\[
	D'_i = \Big( ~a'_j=d_{\Hout}(x_i,x_j) ~\Big)_{j\in [k]}.
\]
The objective now is to maintain the sequence of values
\[
	\tilde D_i = \Big( ~\tilde a_j=d_{\Hout}(x_i,x_j)-d_H(x_i,x_j)=a'_i-a_i ~\Big)_{j\in [k]}.
\]
as we move the root $x_i$ along $C$.
We start computing $\tilde D_1$ by using a shortest-path tree in $H$ and in $\Hout$
from $x_1$. We store the sequence $\tilde D_1$ using the data structure of 
\cref{le:datastructure-numbers}.

For $i=2$ to $i=k$, to move from $\tilde D_{i-1}$ to $\tilde D_i$
we have to perform the operations of $\Lambda'_i$ (as they are) 
and the operations of $\Lambda_i$ with reversed sign.
Each operation is done in $O(\log n)$ time per
operation using the data structure of \cref{le:datastructure-numbers}.
Whenever we obtain the sequence
$\tilde D_i$, we query the data structure with $\DSmin$ to
obtain 
\[
	m_i ~:=~ \min \tilde D_i  ~=~ \min\{ d_{\Hout}(x_i,x_j)- d_H(x_i,x_j)\mid j\in [k]\}.
\]
If some $m_i$ is negative, then we know that for some $x_j\in V(C)$ we have
\[
	d_{\Hout}(x_i,x_j) < d_H(x_i,x_j),
\]
and therefore $H$ is not an isometric subgraph of $G$.
If $m_i\ge 0$ for all $i\in [k]$, then we have
\[
	\forall i,j\in [k]: ~~~ d_{\Hout}(x_i,x_j) \ge d_H(x_i,x_j).
\]
Because $\partial H \subseteq V(C) =\{ x_1,\dots,x_k\}$,
it follows from \cref{le:isometric-boundary} that $H$
is an isometric subgraph in $G$.

Because of \cref{le:encoding} we spend $O(n\log n)$ time to 
know which $O(n)$ operations in the sequence of numbers have to
be performed. Since each operation in the sequence of numbers
takes $O(\log n)$ time (\cref{le:datastructure-numbers}), 
the total running time is $O(n\log n)$.

A similar procedure can be done for testing \emph{convexity} because of 
\cref{le:convex-boundary}:
the graph $H$ is convex in $G$ if and only if 
\[
	\forall i\in [k]:~~~ 
		\min\{ d_{\Hout}(x_i,x_j)- d_H(x_i,x_j)\mid j\in [k]\setminus \{i\}\} >0.
\]
For any $i\in [k]$, we can check whether 
$\min\{ d_{\Hout}(x_i,x_j)- d_H(x_i,x_j)\mid j\in [k]\setminus \{i\}\}$
is strictly positive using $\tilde D_i$; for example, we can add $1$ to
the value $\big(\tilde D_i\big)_i$ in the sequence $\tilde D_i$ (to make it 
non-zero, as otherwise 
$\big(\tilde D_i\big)_i = d_{\Hout}(x_i,x_i)- d_H(x_i,x_i)=0$), then
query for $\min \tilde D_i$, and then subtract $1$ to $\big(\tilde D_i\big)_i$
to restore $\tilde D_i$.
We summarize.

\begin{theorem}
	Let $G$ be a plane graph with $n$ vertices and let $H=\interior(G,C)$
	be a subgraph of $G$ enclosed by a cycle $C$ of $G$.
	In $O(n\log n)$ time we can test whether $H$ is an isometric subgraph of $G$
	and whether $H$ is a convex subgraph of $G$.
\end{theorem}

The result can be extended to subgraphs of a plane graph that are defined
by a few cycles.
Assume that $C$ is a cycle of $G$ and $\mathcal{D}$ is a family of 
cycles such that each cycle of $\mathcal{D}$ is inside $C$,
and no cycle of $\mathcal{D}$ is inside another cycle of $\mathcal{D}$.
We can then define the graph $H=\Holes(G,C,\mathcal{D})$
as the subgraph induced by the edge set 
\[
	E(\interior(G,C))\setminus \bigcup_{C'\in \mathcal{D}} \Big( E(\interior(G,C'))\setminus E(C')\Big).
\]
The idea is that $H$ is the graph contained in the closure of the region
having the curves defined by $C$ and the cycles of $\mathcal{D}$ on the boundary. 
See \cref{fig:example}. Note that $\interior(G,C)=\Holes(G,C,\emptyset)$.
If $\mathcal{D}$ has a constant number of cycles,
then we can use a similar method for each of the cycles in $\mathcal{D}$.
A simple way to do this is to re-embed the graph $G$ such that the 
cycle under consideration is the one that contains the others.
The method then takes $O(|\mathcal{D}| n\log n)$ time.

To finalize this section, we note that the criteria using the Wiener index
with marked vertices (\cref{le:isometric-sum-boundary,le:convex-sum-boundary})
does not seem to lead to near-linear time in this setting.
If $H=\interior(G,C)$, we can adapt the algorithms of 
\cite{CabelloCE13,DasKGW22,Klein05} to compute $\sum_{x,y\in V(C)}d_H(x,y)$
in $O(n\log n)$ time.
However, we do not know how to compute $\sum_{x,y\in V(C)}d_G(x,y)$
so efficiently.
This is the reason that we have turned our attention to the different criteria
given in \cref{le:isometric-boundary,le:convex-boundary},
where the difference of distances is considered. We then leverage
that the distances between the vertices of $V(C)$ can be maintained efficiently
in $H$ and $G-E(H)$.

\bibliographystyle{plainurl}
\bibliography{biblio}
\end{document}